\documentclass{llncs}

\newcommand{\longversion}[1]{#1}
\newcommand{\shortversion}[1]{}

\usepackage{amsmath,amssymb}
\usepackage{times}
\usepackage{algorithm}
\usepackage{algpseudocode}

\longversion{\usepackage[totalwidth=420pt,totalheight=625pt]{geometry}}
\longversion{
\setlength{\oddsidemargin}{5mm}
\setlength{\evensidemargin}{5mm}
}

\usepackage{tikz}
\usetikzlibrary{shapes}

\usepackage{thmtools,thm-restate}

\tikzstyle{literal}=[circle, thin, draw, minimum height=0.3cm]
\tikzstyle{clause}=[rectangle, thin, draw, minimum width=0.3cm, minimum
height=0.3cm]

 \newcommand{\FFF}{\mathcal{F}}
\newcommand{\GGG}{\mathcal{G}} 
 
 \newcommand{\RRR}{\mathcal{R}}

\newcommand{\Dres}{D^{\text{\normalfont res}}}
\newcommand{\Dmat}{D^{\text{\normalfont mat}}}

 \newcommand{\var}{\mathit{var}}
 \newcommand{\lit}{\mathit{lit}}

\def\hy{\hbox{-}\nobreak\hskip0pt} 
\newcommand{\SB}{\{\,} \newcommand{\SM}{\;{:}\;} \newcommand{\SE}{\,\}}

\newcommand{\Card}[1]{|#1|}

\let\doendproof\endproof
\renewcommand\endproof{~\hfill\qed\doendproof}

\let\doendexample\endexample
\renewcommand\endexample{~\hfill\qed\doendexample}

\begin{document}

\title{Computing Resolution\hy Path Dependencies\\ in Linear Time%
  \thanks{Research supported by the European Research Council (ERC),
    project COMPLEX REASON 239962, and WWTF grant WWTF016.} %
  \thanks{Dedicated to the memory of Marko Samer.}}

\author{Friedrich Slivovsky \and Stefan Szeider}

\institute{
 Institute of Information Systems,
 Vienna University of Technology,
 A-1040 Vienna, Austria
\email{friedrich.slivovsky@tuwien.ac.at,stefan@szeider.net}
}

\maketitle

\begin{abstract}The alternation of existential and universal quantifiers in
  a quantified boolean formula (QBF) generates dependencies among variables
  that must be respected when evaluating the formula. Dependency schemes
  provide a general framework for representing such dependencies. Since it is
  generally intractable to determine dependencies exactly, a set of potential
  dependencies is computed instead, which may include false positives. Among
  the schemes proposed so far, resolution path dependencies introduce the
  fewest spurious dependencies. In this work, we describe an algorithm that
  detects resolution-path dependencies in linear time, resolving a problem
  posed by Van Gelder (CP 2011).
\end{abstract}

\section{Introduction}
Deciding the satisfiability of quantified boolean formulas (QBF) is a
canonical PSPACE-complete problem \cite{StockmeyerMeyer73}. Under standard
complexity theoretic assumptions, that means it is much harder than testing
satisfiability of propositional formulas. The source of this discrepancy can
be found in variable dependencies introduced by the alternation of universal
and existential quantifiers in a QBF. The kind of dependencies we consider can
be illustrated with the following example:
\[\FFF = \forall x \exists y\: (x \vee \neg y) \wedge (\neg x \vee y) \] 
While $\FFF$ is satisfiable, there is no single satisfying assignment to
$y$. Instead, the value of $y$ that satisfies~$\FFF$ \textit{depends} on the
value of $x$.

For formulas in prenex normal form, it is safe to assume that a variable
depends on all variables to its left in the quantifier prefix, but this
assumption may result in a large number of spurious dependencies. More
accurate representations of the dependency structure in a formula can be
exploited for various purposes, and variable dependencies have been studied in
a series of works, including
\cite{AyariBasin02,Biere04,BubeckKleinebuning07,EglyTompitsWoltran02,%
  LonsingBiere2009,LonsingBiere2010,Samer08,SamerSzeider09a,VanGelder11}.

Unfortunately, the problem of computing variable dependencies exactly is
PSPACE-complete \cite{SamerSzeider09a}. In practice one therefore computes an
over\hy approximation of dependencies that may contain false positives. This
leads to a trade-off between tractability and generality.
\begin{sloppypar}
  In a recent paper, Van Gelder \cite{VanGelder11} introduced
  \emph{resolution\hy path dependencies} and argued that they generate fewer
  spurious dependencies than all previously considered notions of variable
  dependency (see Figure~\ref{fig:lattice}).
\end{sloppypar}
\shortversion{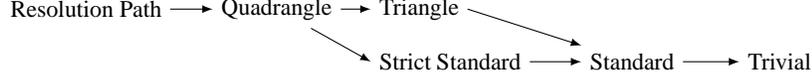
\begin{figure}[b]}
\longversion{\begin{figure}}
\begin{center}
\begin{tikzpicture}[>=latex,auto,scale=0.7]
\node[anchor=west] (RP) at (-7,1) {Resolution Path};
\node[anchor=west] (QD) at (-3,1) {Quadrangle};
\node[anchor=west] (TRI) at (0,1) {Triangle};
\node[anchor=west] (STSTD) at (0,0) {Strict Standard};
\node[anchor=west] (STD) at (4,0) {Standard};
\node[anchor=west] (TRIV) at (7,0) {Trivial};
\path[thin,->] (RP) edge (QD.west);
\path[thin,->] (QD) edge (TRI.west);
\path[thin,->] (QD) edge (STSTD.west);
\path[thin,->] (TRI.east) edge (STD);
\path[thin,->] (STSTD) edge (STD.west);
\path[thin,->] (STD) edge (TRIV.west);
\end{tikzpicture}
\end{center}
\caption{Various notions of variable dependency ordered by generality
  \cite{VanGelder11}.  An arrow from A to B should be read as ``A is strictly
  more general than B.''  \emph{Trivial dependencies} include all pairs of
  variables not contained in the same quantifier block as dependent and serve
  as a baseline. \emph{Standard dependencies} \cite{SamerSzeider09a} identify
  dependencies based on a notion of local connectivity of clauses, extending
  ideas introduced in work on universal expansion
  \cite{Biere04,BubeckKleinebuning07}. \emph{Triangle dependencies} generalize
  standard dependencies without increasing the worst-case asymptotic runtime
  \cite{SamerSzeider09a}. \emph{Quadrangle dependencies} in turn refine
  triangle dependencies, and \emph{strict standard dependencies} refine
  standard dependencies \cite{VanGelder11}. \emph{Resolution path
    dependencies} are based on a sophisticated notion of connectivity motivated
  by properties of Q\hy resolution \cite{VanGelder11}.}
\label{fig:lattice}
\end{figure}
Van Gelder stated as an open problem whether resolution\hy path dependencies
can be computed in polynomial time \cite{VanGelder11}. In this work, we solve
this problem by describing a linear\hy time algorithm that identifies
resolution\hy path dependencies.  We obtain this result by a reduction to the
problem of finding properly colored walks in edge\hy colored graphs, which is
in turn solved using a variant of breadth\hy first search.  We thus show that
the most general dependency relation among those considered so far is
tractable. 

\emph{Dependency schemes} are a generic framework for representing variable
dependencies \cite{SamerSzeider09a} that are useful in various settings. In
particular, they have recently been built into state\hy of\hy the\hy art QBF
solvers, with beneficial effects \cite{LonsingBiere2009,LonsingBiere2010}. We
prove that resolution\hy path dependencies give rise to a dependency scheme,
thereby providing a basis for their use across a variety of applications.

\shortversion{The proofs of statements marked with $(\star)$ have been omitted
  due to space constraints. They can be found in the full version of this
  paper, which is available on arXiv:1202.3097.}
\section{Preliminaries}
\subsection{Quantified Boolean Formulas}
In this section, we cover basic definitions and notation used
throughout the paper. For an in\hy depth treatment of theoretical and
practical aspects of QBFs, we refer the reader to
\cite{KleineBuningBubeck09} and \cite{GiunchigliaMarinNarizzano09},
respectively.

We consider quantified boolean formulas in \emph{quantified
  conjunctive normal form} (QCNF). A QCNF formula consists of a
(quantifier) \emph{prefix} and a CNF formula, called the
\emph{matrix}. A CNF formula is a finite conjunction of
\emph{clauses}, where each clause is a finite disjunction of
\emph{literals}. We identify a CNF formula with the set of its
clauses, and a clause with the set of its literals. Literals are
negated or unnegated propositional \emph{variables}. If $x$ is a
variable, we put $\overline{x} = \neg x$ and $\overline{\neg x} = x$,
and let $\var(x) = \var(\neg x) = x$. If $X$ is a set of literals, we
write $\overline{X}$ for the set $\SB \overline{x} \SM x \in X
\SE$. For a clause $C$, we let $\var(C)$ be the set of variables
occuring (negated or unnegated) in $C$.  For a QCNF formula $\FFF$
with matrix $F$, we put $\var(\FFF) = \var(F) = \bigcup_{C \in
  F}\var(C)$, and $\lit(\FFF) = \var(\FFF) \cup
\overline{\var(\FFF)}$. We call a clause \emph{tautological} if it
contains the same variable negated as well as unnegated. Unless
otherwise stated, we assume that the matrix of a formula does not
contain tautological clauses (tautological clauses can be deleted
without changing satisfiability of a formula). The prefix of a QCNF
formula $\FFF$ is a sequence $\mathsf{Q}_1x_1\dots\mathsf{Q}_nx_n$ of
\emph{quantifications} $\mathsf{Q}_ix_i$, where $x_1, \dots, x_n$ are
pairwise distinct variables in $\var(\FFF)$ and $\mathsf{Q}_i \in
\{\forall, \exists\}$ for $1 \leq i \leq n$. We define the
\emph{depth} of variable $x_p$ as $\delta_{\FFF}(x_p) = p$, and let
$q_{\FFF}(x_p) = \mathsf{Q}_p$. A QCNF formula $\FFF'$ is obtained
from $\FFF$ by \textit{quantifier reordering} if there is a
permutation $i_1,\dots,i_n$ of $1,\dots,n$ such that $\FFF' =
\mathsf{Q}_{i_1}x_{i_1},\dots,\mathsf{Q}_{i_n}x_{i_n} F$, where $F$
denotes the matrix of $\FFF$.
\longversion{\begin{sloppypar}}
The sets of \emph{existential} and \emph{universal} variables occurring
in $\FFF$ are given by $\var_{\exists}(\FFF) = \SB x~\in~\var(\FFF)
\SM q_{\FFF}(x) = \exists \SE$ and $\var_{\forall}(\FFF) = \SB x~\in
\var(\FFF) \SM q_{\FFF}(x) = \forall \SE$, respectively. We call a
literal $\ell$ existential (universal) if $\var(\ell)$ is existential
(universal). We assume that every variable in $\var(\FFF)$ appears in
the prefix of $\FFF$, and -- conversely -- that every variable
quantified in the prefix appears in $F$. The \emph{size} of a QCNF
formula $\FFF$ with matrix $F$ is defined as $\Card{\FFF} = \sum_{C
  \in F}\Card{C}$.

For a set $X$ of variables, a \emph{truth assignment} is a mapping $\tau: X
\rightarrow \{0,1\}$. We extend $\tau$ to literals by setting $\tau(\neg x) =
1 - \tau(x)$, for $x \in X$. Let $\tau: X \rightarrow \{0,1\}$ be a truth
assignment and $F$ a CNF formula. By $F[\tau]$ we denote the formula obtained
from $F$ by removing all clauses containing a literal $\ell$ such that
$\tau(\ell) = 1$, and removing from every clause all literals $\ell$ for which
$\tau(\ell) = 0$; moreover, if $\FFF$ is a QCNF formula, we write $\FFF[\tau]$
for the formula obtained from $\FFF$ by replacing its matrix $F$ with
$F[\tau]$ and deleting all superfluous quantifications in its prefix.

The evaluation function $\nu$ on QCNF formulas is recursively defined
by $\nu(\exists x \FFF) = \max(\nu(\FFF[x \mapsto~0]),\nu(\FFF[x
\mapsto 1]))$, $\nu(\forall x\FFF) = \min(\nu(\FFF[x \mapsto
0]),\nu(\FFF[x \mapsto 1]))$, $\nu(\emptyset) = 1$, and
$\nu(\{\emptyset\}) = 0$, where $x \mapsto \varepsilon$ denotes the
assignment $\tau: \{x\} \rightarrow \{0,1\}$ such that $\tau(x) =
\varepsilon$. A QCNF formula $\FFF$ is \emph{satisfiable} if
$\nu(\FFF) = 1$ and \emph{unsatisfiable} if $\nu(\FFF) = 0$. Two
formulas $\FFF$ and $\FFF'$ are \emph{equivalent} if $\nu(\FFF) =
\nu(\FFF')$.
\longversion{\end{sloppypar}}
We call a clause \emph{ternary} if it contains at most three literals. A QCNF
formula is ternary if all of the clauses in its matrix are ternary. We denote
the class of ternary QCNF formulas by Q3CNF.

\subsection{Q\hy Resolution}
Q\hy resolution \cite{KleinebuningKarpinskiFlogel95} is an extension of
propositional resolution. Let $\FFF$ be QCNF formula with matrix~$F$. A
\emph{tree-like Q\hy resolution derivation} of clause $D$ from $\FFF$ is a
pair $\pi = (T,\lambda)$ of a rooted binary tree $T$ and a labeling $\lambda$
satisfying the following properties. The labeling $\lambda$ assigns to each
node a clause, and to each edge a variable. The leaves of $T$ are labeled with
clauses of $F$, and the root of $T$ is labeled with $D$. Whenever a node $t$
has two children $t'$ and $t''$, then there is an existential literal $\ell$
such that $\ell \in \lambda(t')$, $\overline{\ell} \in \lambda(t'')$, and
$\lambda(tt') = \lambda(tt'') = \var(\ell)$. Moreover, $\lambda(t) =
(\lambda(t') \setminus \{\ell\}) \cup (\lambda(t'')
\setminus\{\overline{\ell}\})$ and $\lambda(t)$ is non-tautological. We call
$\lambda(t)$ the \emph{(Q-)resolvent} of $\lambda(t')$ and $\lambda(t'')$, and
say that $\lambda(t)$ is obtained by \emph{resolution} of $\lambda(t')$ and
$\lambda(t'')$ on variable $\var(\ell)$. If a node $t$ has a single child
$t'$, then $\lambda(t) = \lambda(t')\setminus \{\ell\}$ and $\lambda(tt') =
\var(\ell)$ for some \emph{tailing} universal literal $\ell$ in
$\lambda(t')$. A universal literal $\ell$ is tailing in $\lambda(t')$ if for
all existential variables $x \in \var(\lambda(t'))$, we have $\delta_{\FFF}(x)
< \delta_{\FFF}(\var(\ell))$. The clause $\lambda(t)$ is the result of
\emph{universal reduction} of $\lambda(t')$ on variable $\var(\ell)$. We call
an instance of resolution or universal reduction in $\pi$ a \emph{derivation
  step} in $\pi$. We say $\pi$ is \emph{strict} if for every path $t_1,\dots,
t_n$ from the root of $T$ to one of its leaves we have
$\delta_{\FFF}(\lambda(t_it_{i+1})) < \delta_{\FFF}(\lambda(t_{i+1}t_{i+2}))$,
for all $i \in \{1,\dots,n - 2\}$.  We call $\pi$ \emph{regular} if every
existential variable appears at most once as an edge-label on a path from the
root of $T$ to one of its leaves. For a tree\hy like Q\hy resolution
derivation $\pi = (T,\lambda)$, we define the set of \emph{resolved variables
  of} $\pi$ as $\mathit{resvar}(\pi) = \SB y \in \var_{\exists}(\FFF) \SM$
there is an edge $e \in T$ such that $\lambda(e) = y \SE$. We define the
\emph{height} of a tree\hy like Q\hy resolution derivation $\pi = (T,\lambda)$
as the height of $T$. A tree\hy like Q\hy resolution derivation of the empty
clause from $\FFF$ is called a \emph{Q\hy resolution refutation} of $\FFF$.
\begin{theorem}\label{ThmStrictQres}A QCNF formula $\FFF$ is unsatisfiable if and only if it has a
  strict, tree-like Q\hy resolution refutation.
\end{theorem}
\begin{proof}
  Completeness of ``ordinary'' Q-resolution is proved in
  \cite{KleinebuningKarpinskiFlogel95}. It is straightforward to turn the
  derivations used in this proof into strict, tree\hy like derivations.
\end{proof}

\section{Dependency Schemes}
For a binary relation~$\RRR$ over some set~$V$ we write $\RRR^*$ to denote the
reflexive and transitive \emph{closure} of~$\RRR$, i.e., the smallest
set~$\RRR^*$ such that~$\RRR^* = \RRR \cup \{ (x,x) : x \in V \} \cup \{ (x,y)
: \exists z\;\mbox{such that}\; (x,z) \in \RRR^* \;\mbox{and}\; (z,y) \in \RRR
\}$. Moreover, we let $\RRR(x) = \{y : (x,y) \in \RRR \}$ for~$x \in V$, and
$\RRR(X) = \bigcup_{x \in X} \RRR(x)$ for $X \subseteq V$. For a QCNF formula
$\FFF$, we define the binary relation $R_\FFF$ over~$\var(\FFF)$ as $R_{\FFF}
= \SB (x,y) \SM x,y \in \var(\FFF), \;\delta_{\FFF}(x) < \delta_{\FFF}(y)
\SE$. That is to say, $R_\FFF$ assigns to each variable~$x$ the variables on
the right of~$x$ in the prefix.

\begin{definition}[Shifting]
  Let~$\FFF$ be a QCNF formula and $X \subseteq \var(\FFF)$. We say the QCNF
  formula~$\FFF'$ is obtained from~$\FFF$ by~\emph{down-shifting}~$X$, in
  symbols $\FFF' = S^{\downarrow}(\FFF,X)$, if $\FFF'$~is obtained from~$\FFF$
  by quantifier reordering such that the following conditions hold:
  \begin{enumerate}
  \item $X = R_{\mathcal{F}'}(x)$ for some $x \in \var(\mathcal{F}) =
    \var(\mathcal{F}')$.
  \item $\delta_{\mathcal{F}'}(x) < \delta_{\mathcal{F}'}(y)$ if and
    only if $\delta_{\mathcal{F}}(x) < \delta_{\mathcal{F}}(y)$ for
    all~$x, y \in X$.
  \item $\delta_{\mathcal{F}'}(x) < \delta_{\mathcal{F}'}(y)$ if and only if
    $\delta_{\mathcal{F}}(x) < \delta_{\mathcal{F}}(y)$ for all~$x, y \in
    \var(\mathcal{F}) \setminus X$.
  \end{enumerate}
\end{definition}
For example, let $\FFF = \exists x \forall y \exists z \forall u
\forall w\: F$, and $X = \{x,z,u\}$. Then $S^{\downarrow}(\FFF,X) =
\forall y \forall w \exists x \exists z \forall u\: F$. Note that the
result of shifting is unique. In general, shifting does not yield an
equivalent formula.
\begin{definition}[Dependency scheme]
  A \emph{dependency scheme}~$D$ assigns to each QCNF formula~$\FFF$ a binary
  relation~$D_{\FFF} \subseteq R_{\FFF}$ such that $\FFF$
  and~$S^{\downarrow}(\FFF, D^*_{\FFF}(x))$ are equivalent for all $x \in
  \var(\FFF)$.  A dependency scheme~$D$ is \emph{tractable} if $D_{\FFF}$ can
  be computed in time that is polynomial in~$\Card{\FFF}$.
\end{definition}
Intuitively, for a QCNF formula $\FFF$, variable $x \in \var(\FFF)$,
and dependency scheme~$D$, the set $D_{\FFF}(x)$ consists of variables
that \emph{may} depend on $x$. More specifically, if we want to
simplify $\FFF$ by moving the variable $x$ to the rightmost position
in the prefix, we can use a dependency scheme to identify a set $X$ so
that down-shifting of $X \cup \{x\}$ preserves
satisfiability. Typically, we are interested in dependency schemes
that allow us to identify sound shifts for entire sets of variables.
\begin{definition}[Cumulative]
  \label{defn:cumul}
  A dependency scheme~$D$ is \emph{cumulative} if for every QCNF
  formula~$\mathcal{F}$ and set~$X \subseteq \var(\mathcal{F})$,
  $\mathcal{F}$ and~$S^\downarrow(\mathcal{F}, D^*_{\mathcal{F}}(X))$
  are equivalent.
\end{definition}
Cumulative dependency schemes play a crucial role in the context of
backdoor sets \cite{SamerSzeider09a}, and have been integrated in
search\hy based QBF solvers \cite{LonsingBiere2010}.

It is easy to verify that we can transpose adjacent quantifications
$\mathsf{Q}_x x \mathsf{Q}_yy$ in the prefix of a QCNF $\FFF$ as long as $y
\notin D_{\FFF}(x)$ for some dependency scheme~$D$. In other words, every
dependency scheme satisfies the property defined below.
\begin{definition}[Sound for transpositions]
  \label{def:transsnd} Let $D$ be a function that assigns to each QCNF formula
  $\FFF$ a binary relation $D_{\FFF} \subseteq R_{\FFF}$. We say $D$ is
  \emph{sound for transpositions} if any two QCNF formulas $\FFF =
  \mathsf{Q}_1x_1\dots$\hskip 0pt$\mathsf{Q}_rx_r
  \mathsf{Q}_{r+1}x_{r+1}$\hskip 0pt$ \dots \mathsf{Q}_nx_n F$ and
  $\mathsf{Q}_1x_1\dots$\hskip 0pt$\mathsf{Q}_{r+1}x_{r+1}$\hskip 0pt$
  \mathsf{Q}_rx_r$\hskip 0pt$ \dots\mathsf{Q}_nx_n F$ are equivalent given
  that $(x_r,x_{r+1}) \notin D_{\FFF}$.
\end{definition}
Further restrictions are required when going beyond individual transpositions:
let $\FFF = \forall x \exists y \exists z\:F$, where $F$ is the CNF encoding
of $z \leftrightarrow (x \vee y)$, and let $D$ be a mapping such that
$D(\FFF) = D_{\FFF} = \emptyset$ and $D(\FFF') = R_{\FFF'}$ for $\FFF' \neq
\FFF$. $\FFF$ is satisfiable and remains satisfiable after transposing $y$ and
$x$ (or $y$ and $z$) in the prefix. However, the formula $S^{\downarrow}(\FFF,
D^*_{\FFF}(x)) = \exists y \exists z \forall x\:F$ is unsatisfiable. So $D$ is
sound for transpositions but not a dependency scheme.
\begin{definition}[Continuous]
  \label{def:cont} Let $D$ be a function that maps each QCNF formula $\FFF$ to
  a binary relation $D_{\FFF} \subseteq R_{\FFF}$. We say $D$ is
  \emph{continuous} if the following holds for every pair $\FFF =
  \mathsf{Q}_1x_1\dots$\hskip 0pt$\mathsf{Q}_rx_r
  \mathsf{Q}_{r+1}x_{r+1}$\hskip 0pt$ \dots \mathsf{Q}_nx_n F$ and $\FFF' =
  \mathsf{Q}_1x_1\dots$\hskip 0pt$\mathsf{Q}_{r+1}x_{r+1}$\hskip 0pt$
  \mathsf{Q}_rx_r$\hskip 0pt$ \dots\mathsf{Q}_nx_n F$ of QCNF formulas:
  $D_{\FFF}(v) = D_{\FFF'}(v)$ for $v \in \var(\FFF) \setminus \{ x_r,
  x_{r+1}\}$, and $D_{\FFF'}(x_r) \subseteq D_{\FFF}(x_r)$ as well as
  $D_{\FFF'}(x_{r+1}) \supseteq D_{\FFF}(x_{r+1})$.
\end{definition}
\begin{restatable}{lemma}{LemTransSnd} \label{lem:transsnd}\shortversion{\textup{($\star$)}}
  Let $D$ be a function that maps each QCNF formula $\FFF$ to a binary
  relation $D_{\FFF} \subseteq R_{\FFF}$. If $D$ is sound for transpositions
  and continuous, then $D$ is a cumulative dependency scheme.
\end{restatable}
\longversion{\shortversion{\begin{sloppypar}}
\begin{proof} Choose an arbitrary QCNF formula $\FFF$. Any shift
  $S^{\downarrow}(\FFF,D^*_{\FFF}(X))$ can be represented as a series of
  transpositions of adjacent quantifications $\mathsf{Q}_v v\mathsf{Q}_w w$
  where $v~\in~D^*_{\FFF}(X)$ and $w \notin D^*_{\FFF}(X)$, because the order
  of elements within the sets $D^*_{\FFF}(X)$ and $\var(\FFF) \setminus
  D^*_{\FFF}(X)$ remains unchanged. In other words, there is a sequence
  $\FFF_1,\dots, \FFF_n$ with $\FFF_1 = \FFF$ and $\FFF_n =
  S^{\downarrow}(\FFF,D^*_{\FFF}(X))$, such that for $i~\in~\{1, \dots,
  n-1\}$, the formula $\FFF_{i+1}$ is obtained from $\FFF_{i}$ by transposing
  adjacent quantifications $\mathsf{Q}_{v_i} v_i$ and $\mathsf{Q}_{w_i} w_i$
  in the prefix of $\FFF_{i}$, where $v_i \in D^*_{\FFF}(X)$, $w_i \notin
  D^*_{\FFF}(X)$, and $\delta_{\FFF_i}(w_i) = \delta_{\FFF_i}(v_i) + 1$. We
  prove for all $k \in \{1,\dots,n-1\}$ that $D_{\FFF_k}$ contains no pair
  $(v_j,w_j)$ such that $k \leq j \leq n -1$. By induction on $k$. Because of
  $v_j \in D^*_{\FFF}(X)$ and $w_j \notin D^*_{\FFF}(X)$, we must have $(v_j,
  w_j) \notin D_{\FFF} = D_{\FFF_1}$ for $1 \leq j \leq n - 1$. For the
  induction step, suppose $(v_j,w_j) \notin D_{\FFF_k}$ for all $j \in
  \{k,\dots,n-1\}$, where $1 \leq k \leq n-1$. Since $D$ is continuous, all
  pairs in $D_{\FFF_{k+1}} \setminus D_{\FFF_k}$ must be of the form $(w_k,
  x)$ for some $x \in \var(\FFF_{k+1})$. No such pair can be identical to any
  pair $(v_j, w_j)$, because $w_k \notin D^*_{\FFF}(X)$, while $v_j \in
  D^*_{\FFF}(X)$, where $j \in \{1,\dots, n-1\}$. We conclude that $(v_j,w_j)
  \notin (D_{\FFF_{k+1}} \setminus D_{\FFF_{k}}) \cup D_{\FFF_k} =
  D_{\FFF_{k+1}}$.  The lemma now follows from the fact that $D$ is sound for
  transpositions and $(v_k,w_k)~\notin~D_{\FFF_k}$ for all $k \in \{1,\dots,
  n-1\}$.
\end{proof}
\shortversion{\end{sloppypar}}
}
\begin{restatable}{lemma}{LemSubCum} \label{lem:subcum}\shortversion{\textup{($\star$)}}
  Let $D'$ be a function that maps each QCNF formula $\FFF$ to a binary
  relation $D'_{\FFF} \subseteq R_{\FFF}$, and let $D$ be a cumulative
  dependency scheme. If $D_{\FFF} \subseteq D'_{\FFF}$ for all formulas
  $\FFF$, then $D'$ is a cumulative dependency scheme as well.
\end{restatable}
\longversion{ \longversion{\begin{sloppypar}}
\begin{proof}
  Choose an arbitrary QCNF formula $\FFF$, and let $X \subseteq
  \var(\FFF)$. From $D_{\FFF}\subseteq D'_{\FFF}$ it follows that
  $D'^*_{\FFF}(X) = D^*_{\FFF}(D'^*_{\FFF}(X))$. Since $D$ is a cumulative
  dependency scheme, the formulas $\FFF$ and $S^\downarrow(\FFF,
  D^*_{\FFF}(D'^*_{\FFF}(X))) = S^\downarrow(\FFF, D'^*_{\FFF}(X))$ are
  equivalent.
\end{proof}
\longversion{\end{sloppypar}}
}
\section{Resolution\hy Path Dependencies}
In this section, we will define the resolution path dependency \emph{scheme},
which corresponds to the resolution\hy path dependency relation proposed by
Van Gelder \cite{VanGelder11}. We justify this change of name by proving that
the resolution path dependency scheme is indeed a cumulative dependency
scheme.

Van Gelder \cite{VanGelder11} gives two definitions for resolution paths
(Definitions 4.1 and 5.2), the former being more restrictive than the
latter. The former definition is problematic as we will explain in Example
\ref{ex:nontaut} below. Hence we will base our considerations on the latter
definition, which defines resolution paths as certain walks in a graph
associated with a QBF formula. However, to avoid clashes with graph\hy
theoretic terminology introduced below, we simply define resolution paths as
particular sequences of clauses and literals.

\begin{definition}[Resolution Path] \label{def:rconnected} Let $\FFF$ be a
  QCNF formula with clause set $F$ and $X \subseteq \var_{\exists}(\FFF)$. An
  \emph{$X$\hy resolution path} in $\FFF$ is a sequence of clauses and
  literals $\ell_1,C_1,\ell'_1,\ell_2,C_2,\ell'_2,\dots,$\hskip
  0pt$\ell_n,C_n,\ell'_n$, satisfying the following properties:
  \begin{enumerate}
  \item $C_i \in F$ and $\ell_i, \ell'_i \in \lit(\FFF)$ for $i \in
      \{1, \dots,n\}$.
  \item $\ell_i, \ell'_i \in C_i$ for $i \in \{1, \dots,n\}$.
  \item $\ell_{i+1}= \overline{\ell'_{i}}$ and $\ell'_i, \ell_{i+1} \in X \cup
    \overline{X}$, for $i \in \{1, \dots,n-1\}$.
  \item $\var(\ell_i) \neq \var(\ell'_i)$ for $i \in \{1, \dots,n\}$, and
    $\ell_1 \neq \ell_n'$.
\end{enumerate}
If $\ell_1, \dots, \ell_n'$ is an $X$\hy resolution path in $\FFF$, we say
that $\ell_1$ and $\ell_n'$ are \emph{resolution connected in $\FFF$ with
  respect to $X$}.
\end{definition}

\begin{example}\label{ex:running} \longversion{Consider the following QCNF formula:
\begin{align*}
    \FFF = \exists y_1 \exists y_2 \forall x_1 \exists y_3 \forall x_2
    \underbrace{(x_1 \vee x_2 \vee y_2 \vee y_1)}_{C_1} &\wedge
    \underbrace{(\neg x_1 \vee \neg y_2 \vee \neg y_1 )}_{C_2} \shortversion{\\ &}\wedge
    \underbrace{(\neg y_1 \vee \neg y_3)}_{C_3} \wedge \underbrace{(\neg y_1
      \vee y_3)}_{C_4}
  \end{align*}} 
\shortversion{Let $\FFF = \exists y_1 \exists y_2 \forall x_1 \exists y_3
  \forall x_2\: C_1 \wedge C_2 \wedge C_3 \wedge C_4$, where $C_1 = (x_1 \vee
  x_2 \vee y_2 \vee y_1),$ \hskip 0pt $ C_2 = (\neg x_1 \vee \neg y_2 \vee \neg
  y_1 ),$ \hskip 0pt $ C_3= (\neg y_1 \vee \neg y_3)$, and $C_4 = (\neg y_1 \vee
  y_3)$.}

The sequence $x_1,C_1,y_1,\neg y_1,C_4, y_3$ is a $\{ y_1 \}$\hy resolution
path in $\FFF$, and so the literals $x_1$ and $\neg y_3$ are resolution connected
with respect to $\{ y_1 \}$. By contrast, the sequence $\neg x_1, C_2, \neg
y_1, C_3, \neg y_3$ is not a resolution path in $\FFF$, because $\neg y_1$ is
followed by a clause instead of the complementary literal $y_1$.
\end{example}
\emph{Resolution path dependencies} are induced by a pair of resolution paths
that connect the same two variables in reverse polarities:
\begin{definition}[Dependency pair] Let $\FFF$ be a QCNF formula and
  $x, y \in \var(\FFF)$. We say $(x,y)$ is a \textit{resolution\hy
    path dependency pair} in $\FFF$ with respect to $X \subseteq
  \var_{\exists}(\FFF)$ if at least one of the following conditions
  holds:
  \begin{itemize}
  \item $x$ and $y$, as well as $\neg x$ and $\neg y$, are resolution connected in
    $\FFF$ with respect to $X$.
  \item $x$ and $\neg y$, as well as $\neg x$ and $y$, are resolution connected in
    $\FFF$ with respect to $X$.
  \end{itemize}
\end{definition}

\begin{definition}[Resolution\hy path dependency scheme]
  \label{def:resdep} The \textit{resolution\hy path dependency scheme} is a
  mapping $\Dres$ that assigns to each QCNF formula $\FFF$ the relation
  $\Dres_{\FFF} = \SB (x,y) \in R_{\FFF} \SM q_{\FFF}(x) \neq q_{\FFF}(y)$ and
  $(x,y)$ is a resolution\hy path dependency pair in $\FFF$ with respect to
  $R_{\FFF}(x) \setminus (\var_{\forall}(\FFF) \cup \{x,y\}) \SE$.
\end{definition}
In the formula $\FFF$ of Example 1 above, $(y_1,x_1)$ is resolution\hy path
dependency pair with respect to $\emptyset$, and $(x_1,y_3)$ is a
resolution\hy path dependency pair with respect to $\{y_1,y_2\}$. But while
$(y_1,x_1) \in \Dres_{\FFF}$, we have $(x_1,y_3) \notin \Dres_{\FFF}$, because
$\neg x_1$ is not resolution connected in $\FFF$ to either of $y_3$ or $\neg
y_3$ with respect to $R_{\FFF}(x_1) \setminus \{y_3 \} = \emptyset$.

The next lemma will be needed in the proof of Theorem \ref{thm:resdep1} below.
\begin{restatable}[\cite{VanGelder11}]{lemma}{LemClauseToPath}\label{lem:ClauseToPath}\shortversion{\textup{($\star$)}}
  Let $\FFF$ be QCNF formula, $ \ell, \ell' \in \lit(\FFF)$ where $\ell \neq
  \ell'$, and $\pi = (T,\lambda)$ a regular, tree-like Q-resolution derivation
  of a clause $D$ such that $\ell, \ell' \in D$. Then $\ell$ and $\ell'$ are
  resolution connected in $\FFF$ with respect to $\mathit{resvar}(\pi)$.
\end{restatable}
\longversion{\begin{proof} By induction on the height $n$ of $\pi$. For $n = 0$, $D$ must
  already be contained in $F$, and $\ell,D,\ell'$ is an $\emptyset$\hy
  resolution path in $\FFF$. Now assume the lemma holds for all $m \in
  \{1,\dots,n-1\}$. Let $r$ denote the root of $\pi$. We have to consider two
  cases. (1) If $r$ has a single child $t$, then $\lambda(r) = D$ is the
  result of universal reduction of $\lambda(t)$, and $\lambda(t)$ must already
  contain $\ell$ and $\ell'$. Let $\pi' = (T',\lambda)$, where $T'$ is the
  subtree of $T$ rooted at $t$. Evidently, $\pi'$ is a regular, tree\hy like
  Q\hy resolution derivation of $\lambda(t)$ whose height is strictly smaller
  than that of $\pi$, so we can apply the induction hypothesis and conclude
  that $\ell$ and $\ell$ are resolution connected in $\FFF$ with respect to
  $\mathit{resvar}(\pi') \subseteq \mathit{resvar}(\pi)$.  (2) Suppose $r$ has
  two child nodes $t'$ and $t''$. Then $\lambda(r) = D$ is the resolvent of
  $\lambda(t')$ and $\lambda(t'')$ on some variable $v = \lambda(rt') =
  \lambda(rt'')$. Let $T'$ and $T''$ denote the subtrees of $T$ rooted at $t'$
  and $t''$, respectively, and set $\pi' = (T',\lambda)$, $\pi'' =
  (T'',\lambda)$. If $\ell, \ell' \in \lambda(t')$ or $\ell, \ell' \in
  \lambda(t'')$, we can apply the same reasoning as in case (1). So assume,
  without loss of generality, that $\ell,v \in \lambda(t')$ and $\neg v,\ell'
  \in \lambda(t'')$. Since $\pi'$ and $\pi''$ are regular and tree-like and
  have height at most $n - 1$, we can conclude from the induction hypothesis
  that $\ell$ and $v$ must be resolution connected in $\FFF$ with respect to
  $\mathit{resvar}(\pi')$, and that $\neg v$ and $\ell'$ must be resolution
  connected in $\FFF$ with respect to $\mathit{resvar}(\pi'')$. That means
  there must be a $\mathit{resvar}(\pi')$\hy resolution path $p' =
  \ell,C_1,\ell_1',\dots,\ell_n,C_n,v$, as well as a
  $\mathit{resvar}(\pi'')$\hy resolution path $p'' = \neg
  v,C'_1,\jmath_1',\dots,\jmath_n,C'_n,\ell'$ in $\FFF$. We claim that the
  sequence $p = \ell,C_1,\ell_1',\dots,\ell_n,C_n,v,$\hskip 0pt$\neg
  v,C'_1,\jmath_1',\dots,\jmath_n,C'_n,\ell'$ is a $\mathit{resvar}(\pi)$\hy
  resolution path between $\ell$ and $\ell'$: it is easy to check that
  properties 1-3 of Definition \ref{def:rconnected} are satisfied by $p$
  because they are satisfied by $p'$ and $p''$ individually. Since $\pi$ is
  regular, we must have $v \notin
  \mathit{resvar}(\pi')~\cup~\mathit{resvar}(\pi'')$, and so $p$ has property
  4 as well. It follows that $\ell$ and $\ell'$ are resolution connected in
  $\FFF$ with respect to $\mathit{resvar}(\pi)$.
\end{proof}
} The following result corresponds to
Theorem 4.7 in \cite{VanGelder11}.  \longversion{We were unable to follow the
  proof presented there without assuming that Q\hy resolution derivations are
  strict, so we include our own version below.}
\begin{restatable}[\cite{VanGelder11}]{theorem}{ThmResDep}\label{thm:resdep1}\shortversion{\textup{($\star$)}}
  Let $\FFF$ be a QCNF formula where $\forall u$ is followed by $\exists e$ in
  the quantifier prefix, so that $\delta_{\FFF}(e) = \delta_{\FFF}(u) +
  1$. Suppose $(u,e) \notin \Dres_{\FFF}$. Let $\FFF'$ be the result of
  transposing $\exists e$ and $\forall u$ in the quantifier prefix. Then
  $\FFF'$ and $\FFF$ are equivalent.
\end{restatable}
\longversion{\begin{proof} 
  \shortversion{\begin{sloppypar}} It is sufficient to show that for all truth
    assignments $\tau$ with domain $\SB x~\in~\var(\FFF) \SM \delta_{\FFF}(x)
    < \delta_{\FFF}(u) \SE$, the formula $\FFF[\tau]$ has a Q\hy resolution
    refutation if and only if $\FFF'[\tau]$ has a Q\hy resolution
    refutation. Note that $(u,e) \in \Dres_{\FFF[\tau]}$ implies $(u, e) \in
    \Dres_{\FFF}$ because every resolution path in $\FFF[\tau]$ is a
    resolution path in $\FFF$.  
  \shortversion{\end{sloppypar}} Suppose
  $\FFF'[\tau]$ is unsatisfiable, and let $\pi'$ be a strict, tree\hy like
  Q\hy resolution refutation of $\FFF'[\tau]$. The only derivation step
  admissible in $\pi'$ that cannot occur in a refutation of $\FFF[\tau]$ is
  universal reduction on $u$ of a clause that contains $e$ or $\neg e$. If
  $\pi'$ contains no such step, $\pi'$ is already a refutation of $\FFF[\tau]$
  and we are done. Otherwise, suppose universal reduction on $u$ is applied to
  a clause $C \supseteq \{\ell_e,\ell_u\}$ in $\pi'$, where $\ell_e \in
  \{e,\neg e\}$ and $\ell_u \in \{u, \neg u\}$. We will construct a strict,
  tree\hy like Q\hy resolution refutation $\pi$ that contains one less
  application of universal reduction on $u$ of a clause that contains $e$ or
  $\neg e$. The literal $\ell_u$ is tailing in $C$, so $C$ does not contain
  existential literals of depth greater than
  $\delta_{\FFF'[\tau]}(u)$. Without loss of generality, we can further assume
  that $C$ contains no universal literals of depth greater than
  $\delta_{\FFF'[\tau]}(u)$, so $C = \SB \ell_u, \ell_e \SE$. Let $C' = \{
  \ell_e \}$ be the result of universal reduction of $C$ on $u$, and $D
  \supseteq \{\overline{\ell_e}\} $ be the clause $C'$ is resolved with in
  $\pi'$. Since $\pi'$ is strict, $D$ cannot contain any existential literal
  other than $\overline{\ell_e}$, and we can again assume that there are no
  universal literals in $D$ of depth greater than
  $\delta_{\FFF'[\tau]}(u)$. Moreover, we cannot have $\overline{\ell_u}
  \notin D$. Otherwise -- since every strict, tree\hy like resolution
  refutation is regular -- we could apply Lemma~\ref{lem:ClauseToPath} to
  obtain $(e,u) \in \Dres_{\FFF'}$ and thus $(u,e) \in \Dres_{\FFF}$, a
  contradiction. That is, we either have $D = \{\overline{\ell_e}, \ell_u\}$
  or $D = \{\overline{\ell_e}\}$. To obtain $\pi$, we first resolve $C$ and
  $D$ to derive $\{\ell_u\}$, and then apply universal reduction. Since $\pi'$
  is strict and tree\hy like, it is easily verified that $\pi$ must be as
  well.

  For the converse direction, observe that every Q\hy resolution refutation of
  $\FFF[\tau]$ is also a refutation of~$\FFF'[\tau]$.
\end{proof}
}
With the next example, we illustrate the importance of allowing consecutive
clauses with a tautological Q\hy resolvent in the definition of resolution
paths.
\begin{example}\label{ex:nontaut} \longversion{Consider the following QCNF formula:
  \[ \GGG = \forall u \exists e \exists v \forall x \exists y \exists z
  \underbrace{(u \vee y)}_{C'_1} \wedge \underbrace{(\neg y \vee \neg x \vee
    v)}_{C'_2} \wedge\hskip 0pt \underbrace{(\neg v \vee x \vee z)}_{C'_3}
  \wedge \underbrace{(\neg z \vee e)}_{C'_4} \wedge \underbrace{(\neg u \vee
    \neg e)}_{C'_5} \]}
\shortversion{Let $\GGG = \forall u \exists e \exists v \forall x \exists y
  \exists z \:C'_1 \wedge C'_2 \wedge C'_3 \wedge C'_4 \wedge C'_5$, where
  $C'_1 = (u \vee y)$, $C'_2 = (\neg y \vee \neg x \vee v)$, $C'_3 = (\neg v
  \vee x \vee z)$, $C'_4 = (\neg z \vee e)$, and $C'_5 = (\neg u \vee \neg
  e)$.}
 Figure \ref{fig:resolutionderivation} shows a Q\hy resolution
derivation of the clause $(u \vee e)$ from $\GGG$. By
Lemma~\ref{lem:ClauseToPath}, there must be a $\{v,y,z\}$-resolution path in
$\GGG$ connecting $u$ and $e$, and indeed it is straightforward to check that
$u,C'_1,y,\neg y,C'_2,v,\neg v, C'_3, z,\neg z, C'_4, e$ is a resolution
path. The literals $\neg u$ and $\neg e$ are trivially resolution connected,
so $(u,e)$ is a resolution path dependency pair with respect to $\{v,y,z\}$,
and $(u,e) \in \Dres_{\FFF}$. This is a genuine dependency: it is easily
verified that switching $\forall u$ and $\exists e$ in the prefix of $\GGG$
results in a formula that is unsatisfiable, while $\GGG$ itself is
satisfiable.

  Note that the clauses $C'_2$ and $C'_3$ do not have a non\hy tautological
  resolvent. All resolution paths in $\GGG$ between $u$ and $e$ lead through
  $C'_2$ and $C'_3$. Consequently, if we would restrict Definition
  \ref{def:rconnected} so as to require consecutive clauses in a resolution
  path to have a non\hy tautological Q\hy resolvent (as in Definition 4.1 of
  \cite{VanGelder11}), $u$ and $e$ would no longer be resolution connected in
  $\GGG$, and $e$ would not be identified as dependent on $u$.
\end{example}

\shortversion{\begin{figure}[b]}
\longversion{\begin{figure}}
\begin{center}
\begin{tikzpicture} [yscale = 0.9,grow=up, level 1/.style={sibling distance=50mm},
 level 2/.style={sibling distance=20mm}, scale=0.8]
  \tikzstyle{edge from parent}=[draw,thick] 
    \node {$u \vee e$}
      child {node {$\neg v \vee e$} 
        child {node {$\neg v \vee x \vee e$} 
          child {node {$\neg z \vee e$}
            edge from parent
            node[below right] {$z$}
          }
          child {node {$\neg v \vee x \vee z$}
            edge from parent
            node[below left] {$z$}
          }
          edge from parent
            node[left] {$x$}
      }
      edge from parent
         node[below right] {$v$}
    }
     child {node {$u \vee v$} 
        child {node {$u \vee \neg x \vee v$} 
          child {node {$\neg y \vee \neg x \vee v$}
            edge from parent
            node[below right] {$y$}
          }
          child {node {$u \vee y$}
            edge from parent
            node[below left] {$y$}
          }
          edge from parent
            node[left] {$x$}
      } edge from parent
            node[below left] {$v$}
      };
\end{tikzpicture}
\end{center}
\caption{Q\hy resolution derivation of $u \vee e$ from $\GGG$} 
\label{fig:resolutionderivation}
\end{figure}
\begin{theorem}\label{thm:DresDep}
  $\Dres$ is a cumulative dependency scheme.
\end{theorem}
\begin{proof} We prove that $\Dres$ is (a) continuous and (b) sound for
  transpositions. The result then follows by Lemma~\ref{lem:transsnd}. (a) Let
  $\FFF$ and $\FFF'$ be QCNF formulas such that $\FFF'$ is obtained from
  $\FFF$ by quantifier reordering. Let $x \in \var(\FFF) = \var(\FFF')$, and
  $P = R_{\FFF}(x) \setminus (\var_{\forall}(\FFF) \cup \{x\})$, $P' =
  R_{\FFF'}(x) \setminus (\var_{\forall}(\FFF') \cup \{x\})$. The set of
  $P$\hy resolution paths in $\FFF$ starting from $x$ is identical to the set
  of $P'$\hy resolution paths in $\FFF'$ starting from $x$ unless $R_{\FFF}(x)
  \neq R_{\FFF'}(x)$. If $R_{\FFF}(x) \subseteq R_{\FFF'}(x)$, every $P$\hy
  resolution path in $\FFF$ is a $P'$\hy resolution path in $\FFF'$. It is an
  easy consequence that $\Dres$ is continuous.

  (b) Let $\FFF$ be a QCNF formula and $x,y \in \var(\FFF)$ so that
  $\delta_{\FFF}(y) = \delta_{\FFF}(x) + 1$ and $(x,y) \notin
  \Dres_{\FFF}$. If $x \in \var_{\forall}(\FFF)$ and $y \in
  \var_{\exists}(\FFF)$, the result follows from
  Theorem~\ref{thm:resdep1}. Suppose $x \in \var_{\exists}(\FFF)$ and $y \in
  \var_{\forall}(\FFF)$. Let $\FFF'$ be the result of transposing $\exists x$
  and $\forall y$ in the quantifier prefix of $\FFF$. Because of $(x,y) \notin
  \Dres_{\FFF}$, we must have $(y,x) \notin \Dres_{\FFF'}$, so we can again
  apply Theorem~\ref{thm:resdep1} and conclude that $\FFF$ and $\FFF'$ are
  equivalent. If $q_{\FFF}(x) = q_{\FFF}(y)$, equivalence is trivial.
\end{proof}
Using Lemma \ref{lem:subcum}, we can conclude that all dependency relations
appearing in Figure \ref{fig:lattice} are cumulative dependency schemes.
\section{Computing Resolution\hy Path Dependencies} \label{sec:tractability}
This section will be devoted to proving that $\Dres$ is tractable. More
specifically, we will show that the set of literals that are resolution
connected to a given literal in a QCNF formula $\FFF$ with respect to a set $X
\subseteq \var_{\exists}(\FFF)$ can be computed in linear time. This result in
turn establishes linear time\hy tractability of deciding whether a pair of
variables is contained in $\Dres_{\FFF}$.

We will reduce the problem of finding resolution paths to the task of finding
properly edge\hy colored walks in certain edge\hy colored graphs. A
\emph{graph} $G$ consists of a finite set $V(G)$ of \emph{vertices} and a set
$E(G)$ of \emph{edges}, where the edge between two vertices $u$ and $v$ is
denoted by $uv$ or equivalently $vu$. All graphs we consider are undirected
and simple (i.e., without self\hy loops or multi\hy edges). If $G$ is a graph
and $v \in V(G)$, elements of the set $N_G(v) = \SB w \in V(G) \SM vw \in E(G)
\SE$ are called \emph{neighbors of $v$ in $G$}. In a \emph{$c$\hy edge\hy
  colored graph} $G$, every edge $e \in E(G)$ is assigned a \emph{color}
$\chi_{G}(e) \in \{ 1, \dots, k \}$. Given a (not necessarily edge\hy colored)
graph $G$, a \emph{walk} from $s$ to $t$ in $G$ is a sequence of vertices $\pi
= v_1,v_2, \dots, v_n$, where $v_1 = s$, $v_n = t$, and $v_iv_{i+1} \in E(G)$
for $i=1,\dots,n-1$. If further $v_i \neq v_{i+2}$ for all $i \in \{1, \dots,
n~-~2 \}$, $\pi$ is said to be \emph{retracting\hy free}. A walk $\pi =
v_1,\dots,v_n$ in a $c$\hy edge\hy colored graph $G$ is \emph{properly edge\hy
  colored (PEC)} if $\chi_G(v_iv_{i+1}) \neq \chi_G(v_{i+1}v_{i+2})$ for all
$i \in \{ 1,\dots, n - 2 \}$. A walk $v_1,\dots,v_n$ satisfying $v_i \neq v_j$
for distinct $i,j \in \{1,\dots,n\}$ is a \emph{path}. A PEC walk which is a
path is called a \emph{PEC path}. The \emph{length} of a
walk $v_1,\dots,v_{n+1}$ is $n$. For $2$\hy edge\hy colored graphs, we use the
names \emph{red} and \emph{blue} to denote the colors $1$ and $2$,
respectively.

Note that there can be a PEC walk from a vertex $s$ to a vertex $t$ without
there being a PEC path from $s$ to $t$. For instance, consider a $2$\hy
edge\hy colored graph with vertex set $\{s,u,v,w,t\}$ and edge set
$\{su,ut,uv,uw,vw\}$, such that $uv$ and $uw$ are red and the remaining edges
are blue. The sequence $s,u,v,w,u,t$ is a PEC walk from $s$ to $t$, but there
is no PEC path from $s$ to $t$.
\subsubsection{Construction.}  Let $\FFF$ be a QCNF formula with matrix
$F$, and let $X \subseteq \var_{\exists}(\FFF)$. We construct two graphs
$G_{\FFF,X}$ and $G'_{\FFF,X}$:
\begin{itemize}
\item For the set of vertices of $G_{\FFF,X}$, we choose $F \cup
  \lit(\FFF)$. Its edge set consists of all edges $\neg zz$ for $z \in X$, and
  all edges $C\ell$ where $\ell \in C$.

\item We define $G'_{\FFF,X}$ to be a $2$\hy edge\hy colored graph with vertex
  set $\lit(\FFF)$ and edge set $E_r \cup E_b$, where the set $E_r$ consists
  of all edges $\neg zz$ for $z \in X$, and $E_b$ consists of all edges
  $\ell\ell'$ such that there is a clause $C \in F$ with $\ell, \ell' \in
  C$. The edges in $E_r$ are red, while those in $E_b$ are
  blue. 
\end{itemize}
For general QCNF formulas $\FFF$, the size of $G'_{\FFF,X}$ can be quadratic
in the size of $\FFF$, since every clause of size $n$ gives rise to a clique
with $n$ vertices. This can be avoided by using the following trick: we first
convert $\FFF$ to a Q3CNF formula $\FFF'$ and then carry out the
construction. For any set $X' \subseteq \var(\FFF')$, we can clearly compute
$G'_{\FFF',X'}$ in time $\mathcal{O}(\Card{\FFF'})$. Furthermore, it is well
known that SAT can be reduced to 3SAT in linear time
\cite{KleineBuningLettman99}. We show that this reduction preserves resolution
connectedness.
\begin{restatable}{lemma}{LemThreeCNF}\label{lem:3cnf}\shortversion{\textup{($\star$)}}
  Let $\FFF$ be an arbitrary QCNF formula and $X \subseteq
  \var_{\exists}(\FFF)$. In time $\mathcal{O}(\Card{\FFF})$, one can construct
  a Q3CNF formula $\FFF'$ and a set $X' \subseteq \var_{\exists}(\FFF')$
  satisfying the following property: two literals $\ell, \ell' \in \lit(\FFF)$
  are resolution connected in $\FFF$ with respect to $X$ if and only if $\ell$
  and $\ell'$ are r\hy connected in $\FFF'$ with respect to $X'$.
\end{restatable}
\longversion{\begin{proof}
  Let $\FFF = \mathsf{Q}_1x_1\dots\mathsf{Q}_rx_r F$, and suppose there is a
  clause $C \in F$ such that $C = (\ell_1 \vee \ell_2 \vee \dots \vee \ell_n)$
  and $n > 3$, where $\ell_i \in \lit(\FFF)$ for all $i \in \{1, \dots, n
  \}$. Let $z$ be a variable not contained in $\var(\FFF)$. We set $\FFF' =
  \exists z \mathsf{Q}_1x_1\dots\mathsf{Q}_rx_r F'$, where $F' = (F \setminus
  \{ C \}) \cup \{C',C''\}$, for $C' = (\ell_1 \vee \ell_2 \vee z)$ and $C'' =
  (\neg z \vee \ell_3 \vee \dots \vee \ell_n)$. We will show that two literals
  $\ell, \ell' \in \lit(\FFF)$ are resolution connected in $\FFF$ with respect
  to $X \subseteq \var_{\exists}(\FFF)$ if and only if $\ell$ and $\ell'$ are
  resolution connected in $\FFF'$ with respect to $X \cup \{z\}$. Let $\ell, \dots,
  \ell_{j_1}, C, \ell_{k_1}, \dots,$\hskip 0pt$\ell_{j_m}, C,
  \ell_{k_m},\dots, \ell'$ be an $X$\hy resolution path in $\FFF$. If $1 \leq
  j_1,k_1,\dots,j_m,k_m < 3$ or $3 \leq j_1,k_1,\dots,j_m,k_m \leq n$, we
  simply replace every occurrence of $C$ with $C'$ or $C''$, respectively, to
  obtain an $X \cup \{z\}$\hy resolution path in $\FFF'$. Without loss of
  generality, suppose $1 \leq j_1,\dots,j_m < 3$ and $3 \leq k_1,\dots,k_m
  \leq n$. It is easy to verify that $\ell,\dots,\ell_{j_1},C',z,\neg
  z,$\hskip 0pt$C'',\ell_{k_1}, \dots,\ell_{j_m},C',z,\neg z,$\hskip
  0pt$C'',\ell_{k_m},\dots, \ell'$ is an $X \cup \{z\}$\hy resolution path in
  $\FFF'$ (recall that $z$ does not occur in anywhere in $\FFF$). For the
  converse, we proceed in the opposite direction, substituting $\ell,C,\ell'$
  for subsequences $\ell,C',z,\neg z, C'',\ell'$ and $\ell,C'',\neg z, z,
  C',\ell'$ of an $X\cup\{z\}$\hy resolution path in $\FFF'$ (where
  $\ell,\ell' \in \lit(\FFF)$). Because $C$ is non\hy tautological, $C'$ and
  $C''$ only have the variable $z$ in common. Keeping this in mind, it is
  straightforward to check that the resulting sequence is an $X$\hy resolution
  path in $\FFF$.

  We obtain the desired Q3CNF formula from a QCNF formula $\FFF$ by scanning
  $\FFF$ from left to right, splitting clauses where necessary. This can be
  done in time $\mathcal{O}(\Card{\FFF})$.
\end{proof}
}
\begin{proposition}\label{prop:linearreduction}
  Given a Q3CNF formula $\FFF$ and a set $X \subseteq \var_{\exists}(\FFF)$,
  the graph $G'_{\FFF,X}$ can be constructed in time
  $\mathcal{O}(\Card{\FFF})$.
\end{proposition}
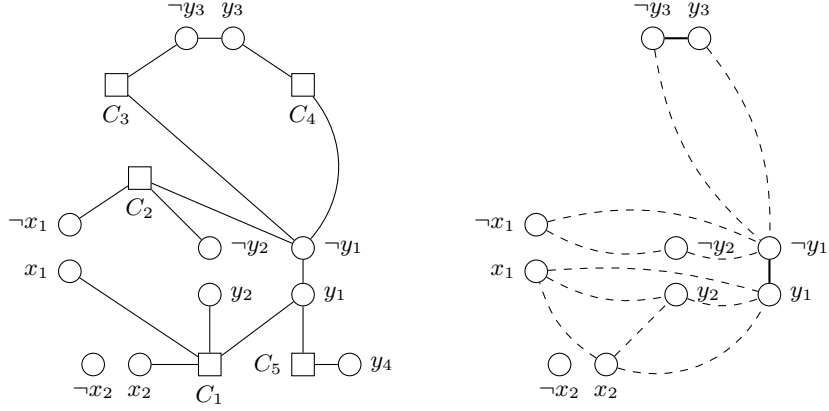
\begin{figure}
\label{fig:clauseliteralgraph}
\shortversion{\vskip -15pt}
\begin{center}
\begin{tikzpicture}[auto, thick, scale=0.62]
\node [clause, label=below:$C_1$] (c1) at (0,0) {};
\node [clause, label=below:$C_2$] (c2) at (-1.5,4) {};

\node [clause, label=below:$C_3$] (c3) at (-2,6) {};

\node [clause, label=below:$C_4$] (c4) at (2,6) {};

\node [literal, label=above:$y_3$] (x) at (0.5,7) {};
\node [literal, label=above:$\neg y_3$] (nx) at (-0.5,7) {};

\node [literal, label=right:$y_2$] (u) at (0,1.5) {};
\node [literal, label=right:$\neg y_2$] (nu) at (0,2.5) {};

\node [literal, label=right:$y_1$] (y) at (2,1.5) {};
\node [literal, label=right:$\neg y_1$] (ny) at (2,2.5) {};

\node [clause, label=left:$C_5$] (c5) at (2,0) {};
\node [literal, label=right:$y_4$] (y4) at (3,0) {};

\node [literal, label=left:$x_1$] (v) at (-3,2) {};
\node [literal, label=left:$\neg x_1$] (nv) at (-3,3) {};

\node [literal, label=below:$x_2$] (z) at (-1.5,0) {};
\node [literal, label=below:$\neg x_2$] (nz) at (-2.5,0) {};

\path[thin] (y) edge (ny);
\path[thin] (x) edge (nx);

\path[thin] (c1) edge (z);
\path[thin] (c1) edge (v);
\path[thin] (c1) edge (u);
\path[thin] (c1) edge (y);

\path[thin] (c2) edge (ny);
\path[thin] (c2) edge (nv);
\path[thin] (c2) edge (nu);

\path[thin] (c3) edge (nx);
\path[thin] (c3) edge (ny);

\path[thin] (c4) edge (x);
\path[thin,bend left=40] (c4) edge (ny);

\path[thin] (c5) edge (y);
\path[thin] (c5) edge (y4);

\node [literal, label=above:$y_3$] (xc) at (10.5,7) {};
\node [literal, label=above:$\neg y_3$] (nxc) at (9.5,7) {};

\node [literal, label=right:$y_2$] (uc) at (10,1.5) {};
\node [literal, label=right:$\neg y_2$] (nuc) at (10,2.5) {};

\node [literal, label=right:$y_1$] (yc) at (12,1.5) {};
\node [literal, label=right:$\neg y_1$] (nyc) at (12,2.5) {};
\path[thick] (yc) edge (nyc);

\node [literal, label=left:$x_1$] (vc) at (7,2) {};
\node [literal, label=left:$\neg x_1$] (nvc) at (7,3) {};

\node [literal, label=below:$x_2$] (zc) at (8.5,0) {};
\node [literal, label=below:$\neg x_2$] (nzc) at (7.5,0) {};

\path[thick] (xc) edge (nxc);
\path[thin,dashed,bend left=20] (xc) edge (nyc);

\path[thin,dashed,bend right=20] (nxc) edge (nyc);

\path[thin,dashed,bend left=20] (nvc) edge (nyc);
\path[thin,dashed,bend right=20] (nvc) edge (nuc);

\path[thin,dashed,bend right=20] (nuc) edge (nyc);

\path[thin,dashed,bend right=40] (zc) edge (yc);
\path[thin,dashed,bend left=20] (zc) edge (vc);
\path[thin,dashed] (zc) edge (uc);

\path[thin,dashed,bend right=20] (vc) edge (uc);
\path[thin,dashed,bend left=10] (vc) edge (yc);

\path[thin,dashed,bend right=20] (uc) edge (yc);

\end{tikzpicture}
\end{center}
\caption{The graphs $G_{\FFF,X}$ (left) and $G'_{\FFF,X}$ (right) for the
  formula $\FFF$ of Example \ref{ex:running} and $X~=~\{y_1,y_3\}$. Red edges
  of $G'_{\FFF,X}$ are represented by solid lines, and blue edges by dashed
  lines.}
\end{figure}
\begin{lemma} \label{lem:reduction} Let $\FFF$ be a QCNF formula, $X \subseteq
  \var_{\exists}(\FFF)$, and $\ell, \ell' \in \lit(\FFF)$ such that $\ell \neq
  \ell'$. The following statements are equivalent:
\begin{enumerate}
\item $\ell$ and $\ell'$ are resolution connected in $\FFF$ with respect to $X$.
\item There is a retracting\hy free walk $\ell_1, C_1,
  \ell'_1,\ell_2,C_2,\ell'_2,\dots, \ell_n,C_n, \ell'_n$ in $G_{\FFF,X}$ from
  $\ell$ to $\ell'$, where $C_i \in F$ and $\ell_i,\ell'_i \in \lit(\FFF)$ for
  $i \in \{1,\dots,n\}$.
\item There is a PEC walk in $G'_{\FFF,X}$ from $\ell$ to
  $\ell'$ whose first and last edges are blue.
\end{enumerate}
\end{lemma}
\begin{proof}
  (1 $\Rightarrow$ 2) Suppose $\ell$ and $\ell'$ are resolution connected in $\FFF$
  with respect to $X$. Then there exists an $X$\hy resolution path $\pi =
  \ell_1,C_1,\ell'_1,\ell_2,C_2,\dots,$\hskip 0pt$\ell_n,C_n,\ell'_n$ in
  $\FFF$ from $\ell$ to $\ell'$. We claim that $\pi$ is already a
  retracting\hy free walk in $G_{\FFF,X}$ of the desired form. Because $\pi$
  is a resolution path, we have $\ell_{i+1} = \overline{\ell_i'}$ and
  therefore $\ell'_i\ell_{i+1}$ in $E(G_{\FFF, X})$ for all $i \in \{ 1,
  \dots, n-1 \}$. Moreover, because $\ell_i,\ell'_i \in C_i$ for $i \in \{
  1, \dots, n \}$, we have $\ell_iC_i, \ell'_iC_i \in E(G_{\FFF, X})$ as
  well. So $\pi$ is indeed a walk in $G_{\FFF,X}$. Since $\var(\ell_i) \neq
  \var(\ell'_i)$ for $i \in \{ 1, \dots, n \}$, $\pi$ must be retracting\hy
  free.

  (2 $\Rightarrow$ 3) Let $\pi = \ell_1, C_1, \ell'_1, \dots, \ell_n,C_n,
  \ell'_n$ be a retracting\hy free walk from $\ell$ to $\ell'$ in $G_{\FFF,X}$
  so that $C_i \in F$ and $\ell_i, \ell'_i \in \lit(\FFF)$ for $i \in \{ 1,
  \dots, n \}$. We show that the sequence $\pi' = \ell_1, \ell'_1, \dots,
  \ell_n,\ell'_n$ is a PEC walk from $\ell$ to $\ell'$ in
  $G'_{\FFF,X}$ whose first and last edges are blue. Let $\ell_iC_i$,
  $C_i\ell'_i$ be a pair of consecutive edges in $\pi$ where $i \in \{ 1,
  \dots, n \}$. By construction of $G_{\FFF,X}$, we have $\ell_i,\ell'_i \in
  C_i$. Because $\pi$ is retracting\hy free, $\ell_i \neq \ell'_i$, and thus
  there is a blue edge $\ell_i\ell'_i$ in $G'_{\FFF,X}$. For all $i \in \{1,
  \dots, n-1\}$, the edge $\ell'_i\ell_{i+1}$ of $\pi$ is a red edge in
  $G'_{\FFF,X}$. So $\pi'$ is a walk in $G'_{\FFF,X}$. Moreover, the first and
  last edges of $\pi$ are blue, and it is easily to verified that $\pi'$ is
  PEC.

  (3 $\Rightarrow$ 1) Now let $\pi = \ell_1, \ell'_1, \ell_2,\ell'_2, \dots,
  \ell_n, \ell'_n$ be a PEC walk from $\ell$ to $\ell'$ in $G'_{\FFF,X}$ whose
  first and last edges are blue. By construction of $G'_{\FFF,X}$, for every
  blue edge $\ell_i \ell'_i$ traversed by $\pi$, there is a clause $C_i$ in
  $F$ such that $\ell_i,\ell'_i \in C_i$, for $i \in \{ 1, \dots, n \}$. For
  every red edge $\ell'_i\ell_{i+1}$, where $i \in \{ 1, \dots, n-1 \}$, we
  have $\ell_{i+1} = \overline{\ell'_i}$ and $\ell_i',\ell_{i+1} \in X \cup
  \overline{X}$. Let $\pi'$ be the sequence $\ell_1, C_1,
  \ell'_1,\dots,$\hskip 0pt$\ell_n,C_n,\ell'_n$. $\pi'$ is an $X$\hy
  resolution path in $\FFF$: we already know that $\pi'$ satisfies conditions
  1-3 of Definition~\ref{def:rconnected}. To verify condition 4, we must show
  that $\var(\ell_i) \neq \var(\ell'_i)$ for all $i \in \{ 1, \dots, n
  \}$. Suppose to the contrary that $\var(\ell_i) = \var(\ell'_i)$ for some $i
  \in \{1,\dots,n\}$. Because $G'_{\FFF,X}$ does not contain self\hy loops,
  this implies $\ell'_i = \overline{\ell_i}$. But then $\ell_i,
  \overline{\ell_i} \in C_i$, contrary to the assumption that $F$ does not
  contain tautological clauses. This concludes the proof that $\pi'$ is an
  $X$\hy resolution path in $\FFF$. It follows that $\ell$ and $\ell'$ are
  resolution connected in $\FFF$ with respect to $X$.
\end{proof}

\noindent \textbf{Algorithm PEC-Walk.} We now describe the algorithm
\emph{PEC-Walk} that takes as input a 2\hy edge\hy colored graph~$G$ and a
vertex $s \in V(G)$, and computes the set of vertices $t$ such that there is a
PEC walk from $s$ to $t$ whose first and last edges are blue. We maintain a
set $Q$ containing (ordered) pairs of vertices $(v,w)$ joined by edges that
can be traversed by a PEC walk starting from $s$. Initially, $Q$ is empty. For
each vertex $v$, we store a set $\psi(v) \subseteq \{\mathit{red},
\mathit{blue} \}$, where $c \in \psi(v)$ indicates that there is a PEC walk
from $s$ to $v$ ending in an edge with color $c$. In an initialization phase,
we first set $\psi(u) = \emptyset$ for all vertices $u$. We then add all pairs
$(s,v)$ to $Q$ such that $v$ is a neighbor of $s$ and $sv$ is a blue edge,
inserting $\mathit{blue}$ into $\psi(v)$ at the same time. In the main
procedure, we repeat the following steps until $Q$ is empty: we remove a pair
$(v,w)$ from $Q$ and add all pairs $(w,u)$ to $Q$ such that $u$ is a neighbor
of $w$, $wu$ is an edge with color $c$ different from the color of $vw$, and
$c$ is not already in $\psi(w)$. For every pair $(v,w)$ we put into $Q$, we
add its color to $\psi(w)$. \longversion{Pseudocode for the algorithm is shown
  below. \begin{algorithm}
\renewcommand{\thealgorithm}{}
\begin{algorithmic}[1]
\ForAll{$v \in V(G)$} 
\State $\psi(v) \gets \emptyset$
\EndFor
\Statex
\ForAll{$w \in N_G(s)$ with $\chi_G(sw) = \mathit{blue}$} 
\State add $(s,w)$ to $Q$
\State $\psi(w) \gets \{\mathit{blue}\}$
\EndFor
\Statex
\While{$Q \neq \emptyset$}
\State let $(u,v)$ be an element of $Q$
\State remove $(u,v)$ from $Q$
\ForAll{$w \in N_G(v)$ such that $\chi_G(vw) \neq \chi_G(uv)$}
\If {$\chi_G(vw) \notin \psi(w)$}
\State add $\chi_G(vw)$ to $\psi(w)$
\State add $(v,w)$ to $Q$
\EndIf
\EndFor
\EndWhile
\end{algorithmic}
\caption{PEC-Walk(Graph $G$, $s \in V(G)$)}
\end{algorithm}
}
\begin{lemma}\label{lem:linearAlgorithm}
  Let $G$ be a 2\hy edge\hy colored graph and $s \in V(G)$. On input
  $(G,s)$, \mbox{PEC-Walk} runs in time $\mathcal{O}(\Card{E(G)}+\Card{V(G)})$.
\end{lemma}
\begin{proof}
  Every ordered pair of vertices joined by an edge is examined at most twice
  and added to $Q$ at most once. The algorithm terminates when $Q$ is empty,
  and an element is removed from $Q$ in each iteration. Initialization can
  take at most $\mathcal{O}(\Card{E(G)}+\Card{V(G)})$ steps. So the time
  required by the entire algorithm is $\mathcal{O}(\Card{E(G)}+\Card{V(G)})$.
\end{proof}
\begin{lemma} \label{lem:algorithm} Let $G$ be a 2\hy edge\hy colored graph,
  $s,t \in V(G), s \neq t$, and let $\psi$ be a vertex labeling generated by
  running {PEC-Walk} on input $(G,s)$. There is a PEC walk
  from $s$ to $t$ whose first edge is blue and whose last edge has color $c
  \in \{\mathit{red},\mathit{blue}\}$ if and only if $c \in \psi(t)$.
\end{lemma}
\begin{proof} 
  By the preceding lemma, the algorithm always terminates and produces a
  labeling $\psi$.

  ($\Leftarrow$) Let $t$ be a vertex of $G$ different from $s$. We show that
  if $c \in \psi(t)$, there is a PEC walk from $s$ to $t$ whose first edge is
  blue and whose final edge has color $c$. We proceed by induction on the
  number $n$ of times the algorithm enters the main loop with $c \notin
  \psi(t)$. If $n = 0$, color $c$ is added to $\psi(t)$ during the
  initialization phase, so there must be a blue edge $st$. Assume the
  statement holds for all $0 \leq k \leq n$, and $c$ is added to $\psi(t)$ in
  iteration $n+1$. Then there must be a pair $(v,t)$ with $\chi_G(vt) = c$
  which is added to $Q$ in this iteration. That is the case only if a pair
  $(u,v)$ is removed from $Q$ during the same iteration with $\chi_G(uv) =
  c'$, where $c' \neq c$. The pair $(u,v)$ must have been inserted into $Q$
  before iteration $n+1$, at which point $c'$ was added to $\psi(v)$. Applying
  the induction hypothesis, we can conclude there must be a PEC walk from $s$
  to $v$ such that its first edge is blue and its last edge has color $c'$. By
  appending $vt$ to this walk, we obtain a PEC walk from $s$ to $t$ with the
  desired properties.

  ($\Rightarrow$) Suppose there is a PEC walk from $s$ to $t$ whose first edge
  is blue and whose last edge has color $c$. Let $n$ be the smallest integer
  that is the length of such a walk. We will show by induction on $n$ that $c
  \in \psi(t)$. The case $n = 1$ is taken care of by the initialization phase
  of the algorithm. Suppose the statement holds for all $n \in \{1, \dots,
  m\}$. Let $v_0,\dots, v_{m+1}$ be a PEC walk from $s$ to $t$ with the
  property that its first edge is blue and its last edge has color $c$, and
  assume there is no shorter PEC walk with this property. Then
  $v_0,\dots,v_{m}$ is a PEC walk from $s$ to $v_{m}$ so that $v_0v_1$ is
  blue, and $\chi_G(v_{m-1}v_m) = c'$ where $c \neq c'$. There can be no $k <
  m$ such that there is a PEC walk of length $k$ from $s$ to $v_{m}$ whose
  first edge is blue and whose last edge has color $c'$: otherwise, one could
  append $v_{m}v_{m+1}$ to this path to obtain a PEC walk from $s$ to
  $v_{m+1}$ whose initial edge is blue and whose final edge has color $c$ of
  length $k + 1 < m + 1$, a contradiction. We can therefore apply the
  induction hypothesis and conclude that $c' \in \psi(v_{m})$. Let $(w,v_{m})$
  be the pair that was removed from $Q$ in the iteration of the main loop in
  which $c'$ was added to $\psi(v_{m})$. Because $c' \neq c$, in the same
  iteration the pair $(v_{m},v_{m+1})$ must have been added to $Q$ and $c$ put
  into to $\psi(v_{m+1})$, unless already $c \in \psi(v_{m+1})$.
\end{proof}
The next result is immediate from Lemmas \ref{lem:linearAlgorithm} and
\ref{lem:algorithm}.
\begin{proposition}\label{prop:reachable}
  Given a $2$\hy edge\hy colored graph $G$, a vertex $s \in V(G)$, and some $c
  \in \{ \mathit{red}, \mathit{blue} \}$, the set of vertices reachable from
  $s$ along some PEC walk in $G$ whose first edge is blue and whose last edge
  has color $c$ can be computed in time
  $\mathcal{O}(\Card{E(G)}+\Card{V(G)})$.
\end{proposition}
With all the pieces in place, it is now straightforward to prove our main result.
\begin{theorem} \label{thm:tractable} Given a QCNF formula $\FFF$ and a pair
  of variables $x, y \in \var(\FFF)$, one can decide whether $(x,y) \in
  \Dres_{\FFF}$ in time $\mathcal{O}(\Card{\FFF})$. Hence the resolution\hy
  path dependency scheme is tractable.
\end{theorem}
\begin{proof}
  We prove that there is a linear time decision algorithm. We first check
  whether $q_{\FFF}(x) \neq q_{\FFF}(y)$ and $(x,y)$ is in $R_{\FFF}$. Using
  Lemma \ref{lem:3cnf}, we can then in linear time compute a QCNF formula
  $\FFF'$ and a set $R'$ from $\FFF$ and $R_{\FFF}(x) \setminus
  (\var_{\forall}(\FFF) \cup \{x,y\})$ so that two literals are resolution
  connected in $\FFF'$ with respect to $R'$ if and only if they are resolution
  connected in $\FFF$ with respect to $R_{\FFF}(x) \setminus
  (\var_{\forall}(\FFF) \cup \{x,y\})$. We can then construct the graph
  $G'_{\FFF',R'}$ and determine for all pairs $\ell_x,\ell_y$ with $\ell_x \in
  \{x, \neg x\}$ and $\ell_y \in \{y, \neg y\}$ whether there is a properly
  edge\hy colored walk from $\ell_x$ to $\ell_y$ whose first and last edges
  are blue, which by Lemma \ref{lem:reduction} is equivalent to $\ell_x$ and
  $\ell_y$ being resolution connected in $\FFF'$ with respect to $R'$
  (according to Propositions \ref{prop:linearreduction} and
  \ref{prop:reachable}, this can be done in linear time). Using this
  information, it is straightforward to decide whether $(x,y)$ is a
  resolution\hy path dependency pair in $\FFF$ with respect to $R_{\FFF}(x)
  \setminus (\var_{\forall}(\FFF) \cup \{x,y\})$. Each of these steps requires
  linear time, so we need $\mathcal{O}(\Card{\FFF})$ time in total.
\end{proof}

Samer and Szeider \cite{SamerSzeider09a} generalized the notion of a strong
backdoor set from CNF formulas to QCNF formulas, by adding the requirement
that the backdoor set is closed under a cumulative dependency scheme. They
showed that evaluating QCNF formulas is fixed-parameter tractable (fpt) when
parameterized by the size of a smallest strong backdoor set (with respect to
the classes QHORN or Q2CNF) provided that the considered cumulative dependency
scheme is tractable. By Theorems \ref{thm:DresDep} and \ref{thm:tractable},
one can use the resolution path dependency scheme here and thus get an fpt
result that is stronger than the results achieved by using any of the other
dependency schemes appearing in Fig. \ref{fig:lattice}.

For an existentially quantified variable $y$ in a QCNF $\FFF$, the entire
\emph{set} $\Dres_{\FFF}(y)$ can be computed in linear time: we first
determine the sets $D = \SB \ell \in \lit(\FFF) \SM y$ is resolution connected
to $\ell$ in $\FFF$ with respect to $R_{\FFF}(y) \setminus
\var_{\forall}(\FFF) \SE$ and $D_{\neg} = \SB \ell \in \lit(\FFF) \SM \neg y$
is resolution connected to $\ell$ in $\FFF$ with respect to $R_{\FFF}(y)
\setminus \var_{\forall}(\FFF) \SE$ and store them in a data structure that
allows us to decide membership of literals in constant time (say, an
array). To determine $\Dres_{\FFF}(y)$, we simply check for each element $x$
of $R_{\FFF} \cap \var_{\forall}(\FFF)$ whether $x \in D$ and $\neg x \in
D_{\neg}$, or $\neg x \in D$ and $x \in D_{\neg}$.

Unfortunately we cannot use the same approach to compute the set of dependent
variables $\Dres_{\FFF}(x)$ for a universal variable $x \in
\var_{\forall}(\FFF)$. For every existential variable $y \in
\var_{\exists}(\FFF)$, resolution paths that entail $(x,y) \in \Dres_{\FFF}$
cannot contain $y$ or $\neg y$. Hence the relevant resolution paths are
subject to different constraints for each $y$, and it is not sufficient in
general to construct $G'_{\FFF,X}$ for a single set $X$.

\section{Minimal Dependency Schemes}
The fact that the resolution\hy path dependency scheme is the bottom element
of the lattice represented in Figure \ref{fig:lattice} gives reason to wonder
whether it is the most general dependency scheme. However, computing a minimal
dependency scheme is complete for PSPACE \cite{SamerSzeider09a}. Since the
resolution path dependency scheme is tractable, it follows that it cannot be
minimal. Can we instead prove that $\Dres$ is minimal relative to a class of
``natural'' dependency schemes? At the very least, such a class should include
all the dependency schemes considered so far, which have the following feature
in common: whether a pair of variables is considered dependent is determined
almost entirely in terms of the matrix. We use this property to define a
candidate class.
\begin{definition} A dependency scheme $D$ is called a \emph{matrix dependency
    scheme} if it satisfies the following property: Let $\FFF$ and $\FFF'$ be
  QCNF formulas such that $\FFF'$ is obtained from $\FFF$ by quantifier
  reordering. Moreover, let $x \in \var(\FFF)$ such that $R_{\FFF}(x) =
  R_{\FFF'}(x)$. Then for any $y \in \var(\FFF)$, we have $(x,y) \in D_{\FFF}$
  if and only if $(x,y) \in D_{\FFF'}$.
\end{definition}
The next proposition can be easily verified by inspecting Definition
\ref{def:resdep}.
\begin{proposition}
  The resolution\hy path dependency scheme $\Dres$ is a matrix
  dependency scheme.
\end{proposition}
Unfortunately, $\Dres$ is not even the most general matrix dependency
scheme. We now show that there is a cumulative matrix dependency scheme which
is strictly more general than $\Dres$. Let $\FFF$ be an arbitrary QCNF
formula.
\begin{definition} We let $\Dmat: \FFF \mapsto \Dmat_{\FFF}$, where
  $\Dmat_{\FFF} = \SB (x,y) \in R_{\FFF} \SM$ there is a formula $\FFF' =
  \mathsf{Q}_1 x_1 \dots \mathsf{Q}_x x \mathsf{Q}_y y \dots \mathsf{Q}_n
  x_n\:F$ obtained from $\FFF$ by quantifier reordering, such that
  $R_{\FFF}(x) \supseteq R_{\FFF'}(x)$ and $\nu(\FFF') \neq \nu(\FFF'')$,
  where $\FFF''= \mathsf{Q}_1 x_1 \dots$\hskip 0pt$ \mathsf{Q}_y y
  \mathsf{Q}_x x \dots$\hskip 0pt$ \mathsf{Q}_n x_n\:F \SE$.
\end{definition}
\begin{restatable}{proposition}{PropDmatCumulative}\label{prop:dmatcumulative}\shortversion{\textup{($\star$)}}
$\Dmat$ is a cumulative matrix dependency scheme.
\end{restatable}
\longversion{\longversion{\begin{sloppypar}}
\begin{proof}
  It is evident that $\Dmat$ is sound for transpositions because the
  identity permutation is among those quantified over in the definition
  above. Our next goal is to show that $\Dmat$ is continuous. Consider two
  QCNF formulas $\FFF = \mathsf{Q}_1 x_1 \dots \mathsf{Q}_r x_r
  \mathsf{Q}_{r+1} x_{r+1} $\hskip 0pt$\dots \mathsf{Q}_n x_n\:F$ and $\FFF' =
  \mathsf{Q}_1 x_1 \dots \mathsf{Q}_{r+1} x_{r+1} \mathsf{Q}_r x_r \dots
  \mathsf{Q}_n x_n\:F$. For all $i \in \{1, \dots,r-1,r+2,\dots, n\}$, we have
  $R_{\FFF}(x_i) = R_{\FFF'}(x_i)$ and therefore $\Dmat_{\FFF}(x_i) =
  \Dmat_{\FFF'}(x_i)$. Since $R_{\FFF}(x_r) \subseteq R_{\FFF'}(x_r)$, every
  formula $\FFF''$ obtained by quantifier reordering from $\FFF$ such that
  $R_{\FFF} \supseteq R_{\FFF''}$ is one that can be obtained from $\FFF'$ by
  quantifier reordering such that $R_{\FFF'} \supseteq R_{\FFF''}$. From this
  we obtain $\Dmat_{\FFF}(x_r) \subseteq \Dmat_{\FFF'}(x_r)$. A symmetric
  argument yields $\Dmat_{\FFF'}(x_{r+1}) \subseteq
  \Dmat_{\FFF}(x_{r+1})$, which completes the proof that $\Dmat$ is
  continuous. Applying Lemma \ref{lem:transsnd}, we conclude that $\Dmat$ is
  a cumulative dependency scheme. Moreover, $\Dmat$ is clearly a matrix
  dependency scheme.
\end{proof}
\longversion{\end{sloppypar}}
}
\begin{restatable}{proposition}{PropDmatDres}\label{prop:resdepcontained}\shortversion{\textup{($\star$)}}
  For every QCNF formula $\FFF$, the relation $\Dmat_{\FFF}$ is contained in
  $\Dres_{\FFF}$, and containment is strict in some cases.
\end{restatable}
\longversion{\begin{proof}
  Let $\FFF$ be a QCNF formula with $x,y~\in~\var(\FFF)$, and suppose that
  $(x,y) \in \Dmat_{\FFF}$. Then there is a formula $\FFF' = \mathsf{Q}_1
  x_1 \dots \mathsf{Q}_x x \mathsf{Q}_y y \dots \mathsf{Q}_n x_n\:F$ obtained
  from $\FFF$ by quantifier reordering, such that $R_{\FFF}(x) \supseteq
  R_{\FFF'}(x)$, and $\nu(\FFF') \neq \nu(\FFF'')$, where $\FFF''=
  \mathsf{Q}_1 x_1 \dots$\hskip 0pt$ \mathsf{Q}_y y \mathsf{Q}_x x
  \dots$\hskip 0pt$ \mathsf{Q}_n x_n\:F$. Because $\Dres$ is a dependency
  scheme, we must have $(x,y) \in \Dres_{\FFF'}$. Since $R_{\FFF}(x)
  \supseteq R_{\FFF'}(x)$, if two literals $\ell_x \in \{x,\neg x\}$ and
  $\ell_y \in \{y, \neg y\}$ are resolution connected in $\FFF'$ with respect
  to $R_{\FFF'}(x) \setminus (\var_{\forall}(\FFF') \cup \{x,y\})$, then they
  are resolution connected in $\FFF$ with respect to $R_{\FFF}(x) \setminus
  (\var_{\forall}(\FFF) \cup \{x,y\})$. It follows that $(x,y) \in
  \Dres_{\FFF}$.

  To see that $\Dmat$ is strictly contained in $\Dres$, consider the QCNF
  formula~$\GGG$:
  \[\GGG = \forall x_1 \forall x_2 \exists y\: (x_1) \wedge (\neg x_2 \vee y)
  \wedge (x_2 \vee \neg y) \] It is easy to verify that $\GGG$ is
  unsatisfiable and remains unsatisfiable if $\exists y$ and $ \forall x_2$
  switch positions in the prefix. From this we get $(x_2,y) \notin
  \Dmat_{\GGG}$. At the same time, it is straightforward to check that
  $(x_2,y) \in \Dres_{\GGG}$.
\end{proof}
} 

\longversion{Note that the trivially unsatisfiable clause $(x_1)$ in the above
  formula can be replaced by an arbitrary unsatisfiable formula.}

The reduction applied in the proof of the following result
essentially corresponds to the one used by Samer and Szeider to establish
PSPACE\hy hardness of computing minimal dependency schemes
\cite{SamerSzeider09a}.
\begin{restatable}{proposition}{PropPiTwo}\label{prop:piTwo}\shortversion{\textup{($\star$)}}
  Let $\FFF$ be a QCNF formula with matrix $F$ and $x,y \in \var(\FFF)$. The
  problem of deciding whether there exists a matrix dependency scheme $D$ such
  that $(x,y) \notin D_{\FFF}$ is $\Sigma^P_2$\hy hard.
\end{restatable}
\longversion{\begin{proof}
  Let $\GGG = \forall x_1 \dots \forall x_n \exists y_1 \dots \exists y_m\:G$
  be a QCNF formula, and $x,y$ new variables not in $\var(\FFF)$. Further, let
  $F = G \wedge (x \vee \neg y) \wedge (\neg x \vee y)$, and let $\FFF =
  \forall x_1 \dots \forall x_n \exists y_1 \dots \exists y_m \forall x
  \exists y\: F $. We will show that there is a matrix dependency scheme $D$
  such that $(x,y) \notin D_{\FFF}$ if and only if $\GGG$ is unsatisfiable. It
  is well known that deciding unsatisfiability of quantified boolean formulas
  with a $\forall^* \exists^*$\hy prefix is $\Sigma^P_2$\hy complete
  \cite{Stockmeyer76}.

  Suppose there is a matrix dependency scheme $D$ such that $(x,y) \notin
  D_{\FFF}$. Since $D$ is a dependency scheme, $\FFF$ and $\FFF' =
  S^{\downarrow}(\FFF,D^*_{\FFF}(x))$ are equivalent. It is easy to see that
  $\FFF'$ is unsatisfiable, so $\FFF$ must be unsatisfiable as well. Because
  the formula $\forall x \exists y\: (x \vee \neg y) \wedge (\neg x \vee y)$ is
  satisfiable, already $\GGG$ must have been unsatisfiable.

  On the other hand, if there is no matrix dependency scheme $D$ such that
  $(x,y) \notin D_{\FFF}$, then in particular $(x,y) \in \Dmat_{\FFF}$. By
  definition of $\Dmat$, there must be a formula $\FFF' = \dots \forall x
  \exists y\: F \wedge (x \vee \neg y) \wedge (\neg x \vee y)$ obtained from
  $\FFF$ by quantifier reordering so that transposing $\forall x$ and $\exists
  y$ in the prefix of $\FFF'$ results in a formula $\FFF''$ such that
  $\nu(\FFF') \neq \nu(\FFF'')$. Since $\exists y \forall x\: (x \vee \neg y)
  \wedge (\neg x \vee y)$ is unsatisfiable, we can conclude that $\FFF'$ must
  be satisfiable. Downshifting of existential variables cannot turn a
  satisfiable formula into an unsatisfiable one, so $\FFF$ is satisfiable as
  well, and we can conclude that $\GGG$ is satisfiable.
\end{proof}
}

One may object that these considerations do not rule out the possibility that
$\Dres$ is the most general \emph{tractable} matrix dependency scheme. That
this is not the case can be seen from the following simple argument. For any
nonnegative integer $k$, we define a mapping $D^k$ such that for any QCNF
formula $\FFF$ we have $D^k_{\FFF} = \Dmat_{\FFF}$ if $\Card{\FFF} \leq k$,
and $D^k_{\FFF} = \Dres_{\FFF}$ otherwise. As both $\Dmat$ and $\Dres$ are
cumulative matrix dependency schemes and the relevant properties are defined
pointwise, any such function $D^k$ must be a cumulative matrix dependency
scheme as well. Moreover, each scheme $D^k$ is clearly tractable and from the
proof of Proposition~\ref{prop:resdepcontained} we know that $D^k$ is strictly
more general than $\Dres$ for $k \geq 5$.
\section{Conclusion}
We have shown that resolution path dependencies give rise to a cumulative
dependency scheme that can be decided in linear time. While the latter result
is optimal for the decision problem, we see at least two obstacles for an
efficient implementation. First, computing the entire relation $\Dres_{\FFF}$
using our current algorithm requires $\mathcal{O}(\Card{\FFF}^3)$ time, which
is prohibitive for practical purposes. Second, it is unclear whether one can
find succinct representations of the relation $\Dres_{\FFF}$ similar to those
used for the standard dependency scheme \cite{LonsingBiere2009}. We leave this
issues for future work.

To capture the kind of variable dependencies relevant for expansion\hy based
QBF solvers, Samer considered an alternative definition of dependency schemes
based on variable \emph{independence} \cite{Samer08}. It might be interesting
to study resolution path dependencies in this context as well.
\bibliographystyle{abbrv} 

\begin{thebibliography}{10}

\bibitem{AyariBasin02}
A.~Ayari and D.~Basin.
\newblock Qubos: Deciding quantified {B}oolean logic using propositional
  satisfiability solvers.
\newblock In M.~Aagaard and J.~O'Leary, editors, {\em FMCAD 2002}, volume 2517
  of {\em LNCS}, pages 187--201. Springer Verlag,
  2002.

\bibitem{Biere04} A.~Biere.  \newblock Resolve and expand.  \newblock In
  H.~Hoos and D.~G. Mitchell, editors, {\em SAT 2004}, volume 3542 of {\em
    LNCS}, pages 59--70. Springer Verlag, 2005.

\bibitem{BubeckKleinebuning07} U.~Bubeck and H.~Kleine~B\"{u}ning.  \newblock
  Bounded universal expansion for preprocessing {QBF}.  \newblock In
  J.~Marques-Silva and K.~A. Sakallah, editors, {\em SAT 2007}, volume 4501 of
  {\em LNCS}, pages 244--257. Springer Verlag,
  2007.


\bibitem{EglyTompitsWoltran02}
U.~Egly, H.~Tompits, and S.~Woltran.
\newblock On quantifier shifting for quantified {Boolean} formulas.
\newblock In {\em Proc. SAT'02 Workshop on Theory and Applications of
  Quantified Boolean Formulas}, pages 48--61. Informal Proceedings, 2002.

\bibitem{GiunchigliaMarinNarizzano09}
E.~Giunchiglia, P.~Marin, and M.~Narizzano.
\newblock Reasoning with quantified boolean formulas.
\newblock In A.~Biere, M.~Heule, H.~van Maaren, and T.~Walsh, editors, {\em
  Handbook of Satisfiability}, volume 185, pages 761--780. IOS Press, 2009.

\bibitem{KleineBuningBubeck09}
H.~Kleine~B{\"u}ning and U.~Bubeck.
\newblock Theory of quantified boolean formulas.
\newblock In A.~Biere, M.~Heule, H.~van Maaren, and T.~Walsh, editors, {\em
  Handbook of Satisfiability}, chapter~23, pages 735--760. IOS Press, 2009.

\bibitem{KleinebuningKarpinskiFlogel95}
H.~Kleine~B\"uning, M.~Karpinski, and A.~Fl\"ogel.
\newblock Resolution for quantified {Boolean} formulas.
\newblock {\em Information and Computation}, 117(1):12--18, 1995.

\bibitem{KleineBuningLettman99}
H.~Kleine~B{\"u}ning and T.~Lettman.
\newblock {\em Propositional logic: deduction and algorithms}.
\newblock Cambridge University Press, Cambridge, 1999.

\bibitem{LonsingBiere2009}
F.~Lonsing and A.~Biere.
\newblock A compact representation for syntactic dependencies in {QBFs}.
\newblock In O.~Kullmann, editor, {\em SAT 2009}, volume 5584 of {\em LNCS}, pages
  398--411. Springer, 2009.

\bibitem{LonsingBiere2010} F.~Lonsing and A.~Biere.  \newblock Integrating
  dependency schemes in search-based {QBF} solvers.  \newblock In O.~Strichman
  and S.~Szeider, editors, {\em SAT 2010}, volume 6175 of {\em LNCS}, pages
  158--171. Springer, 2010.

\bibitem{Samer08}
M.~Samer.
\newblock Variable dependencies of quantified {CSPs}.
\newblock In I.~Cervesato, H.~Veith, and A.~Voronkov, editors, {\em LPAR 2008},
  volume 5330 of {\em LNCS}, pages 512--527. Springer, 2008.

\bibitem{SamerSzeider09a}
M.~Samer and S.~Szeider.
\newblock Backdoor sets of quantified {B}oolean formulas.
\newblock {\em Journal of Automated Reasoning}, 42(1):77--97, 2009.

\bibitem{Stockmeyer76}
L.~J. Stockmeyer.
\newblock The polynomial-time hierarchy.
\newblock {\em Theoretical Computer Science}, 3(1):1--22, 1976.

\bibitem{StockmeyerMeyer73}
L.~J. Stockmeyer and A.~R. Meyer.
\newblock Word problems requiring exponential time.
\newblock In {\em Proc. Theory of Computing}, pages 1--9. ACM, 1973.

\bibitem{VanGelder11} A.~Van~Gelder.  \newblock Variable independence and
  resolution paths for quantified boolean formulas.  \newblock In J.~Lee,
  editor, {\em CP 2011}, volume 6876 of {\em LNCS}, pages 789--803. Springer,
  2011.

\end{thebibliography}

\end{document}